\documentclass[letterpaper]{article} 
\usepackage{aaai25}  
\usepackage{times}  
\usepackage{helvet}  
\usepackage{courier}  
\usepackage[hyphens]{url}  
\usepackage{graphicx} 
\usepackage{natbib}  
\usepackage{caption} 
\usepackage{color}
\usepackage{comment}
\usepackage{amsmath, amsthm, amssymb}
\usepackage{physics}
\usepackage{multicol}
\usepackage{multirow}
\usepackage{algorithm2e}
\usepackage{enumerate}
\usepackage{booktabs}
\RestyleAlgo{ruled}
\usepackage{float}
\usepackage{tikz}
\usepackage{mathtools}
\theoremstyle{definition}
\newtheorem{thm}{Theorem}
\newtheorem*{thm*}{Theorem}
\newtheorem{lem}{Lemma}
\newtheorem*{lem*}{Lemma}
\newtheorem{dfn}{Definition}
\newtheorem{prop}{Proposition}
\newcommand{\mathtrue}{\text{\textit{true}}}
\newcommand{\mathfalse}{\text{\textit{false}}}

\usepackage{newfloat}
\usepackage{listings}
\DeclareCaptionStyle{ruled}{labelfont=normalfont,labelsep=colon,strut=off} 
\floatstyle{ruled}
\newfloat{listing}{tb}{lst}{}
\floatname{listing}{Listing}

\title{Tensor Decomposition Meets Knowledge Compilation: \\A Study Comparing Tensor Trains with OBDDs}
\author{
    Ryoma Onaka, Kengo Nakamura, Masaaki Nishino, Norihito Yasuda
}
\affiliations{
  NTT Communication Science Laboratories, NTT Corporation,
Kyoto, Japan\\
    \{ryoma.onaka,kengo.nakamura,masaaki.nishino,norihito.yasuda\}@ntt.com
}

\begin{document}

\maketitle


\begin{abstract}
A knowledge compilation map analyzes tractable operations in Boolean function representations and compares their succinctness.
This enables the selection of appropriate representations for different applications.
In the knowledge compilation map, all representation classes are subsets of the negation normal form (NNF). 
However, Boolean functions may be better expressed by a representation that is different from that of the NNF subsets.
In this study, we treat tensor trains as Boolean function representations and analyze their succinctness and tractability. Our study is the first to evaluate the expressiveness of a tensor decomposition method using criteria from knowledge compilation literature. 
Our main results demonstrate that tensor trains are more succinct than ordered binary decision diagrams (OBDDs) and support the same polytime operations as OBDDs. 
Our study broadens their application by providing a theoretical link between tensor decomposition and existing NNF subsets.
\end{abstract}

\section{Introduction}

Knowledge compilation, a technique for handling propositional reasoning at high-speed, has been studied for decades.
The logical formula in propositional logic can be converted into a representation by pre-computation, allowing various queries to be quickly answered on such converted representations.
Typical target representations for knowledge compilation include \emph{Ordered Binary Decision Diagrams (OBDDs)}~\cite{bryant1986graph}, \emph{deterministic Decomposable Negation Normal Form (d-DNNF)}~\cite{Darwiche2001ddnnf}, and \emph{Sentential Decision Diagrams (SDDs)}~\cite{darwiche2011sdd}.
\emph{A Negation Normal Form (NNF)}~\cite{darwiche-dnnf2001} is a generalization of all these approaches and Boolean function representation classes can be considered as adding constraints, such as determinism and structured decomposability, to the NNF.
However, the possibility remains that Boolean functions can be represented more succinctly by focusing on representations other than NNF.

Tensor decomposition represents a tensor as the sum of the products of smaller tensors.
Tensor decomposition resembles knowledge compilation because it succinctly represents a large tensor with a set of smaller tensors.
A tensor train decomposes a high-dimensional tensor into a linear product of lower-dimensional tensors~\cite{oseledets2011tensor}.
Many operations can be performed on a tensor train, such as sum, inner product, and Hadamard product.

As the truth table of an $n$-ary Boolean function can naturally be regarded as an $n$-dimensional binary tensor, an exact tensor decomposition representing such a binary tensor can be regarded as a target representation for knowledge compilation.
Our study measures the performance of a tensor train as a representation of a Boolean function.
For comparison with existing Boolean function representations, we follow the approach of the knowledge compilation map~\cite{knowledge}, which clarifies the succinctness between Boolean function representations and the operations that can be performed on each representation in polytime.
In this study, we analyze the succinctness between OBDDs and tensor trains.
In addition to the popularity of OBDDs, considering that a linear ordering of variables can be introduced into tensor trains such as OBDDs, we argue that a comparison between tensor trains and OBDDs would be beneficial to the community.

The main result of this study is twofold.
First, we show that tensor trains (deemed to be Boolean function representations) are more succinct than OBDDs; that is, 
they can represent any function with at most a polynomial times larger size than OBDDs and a function that cannot be represented by OBDDs with a polynomial times larger size than tensor trains exists.
Second, we prove that tensor trains can perform all operations on the knowledge compilation map~\cite{knowledge} that OBDDs can perform in polynomial time (polytime) with respect to the size of the representation.
These results are interesting because there is usually a trade-off between succinctness and the abundance of operations that can be performed in polytime, as suggested by knowledge compilation maps.
Moreover, 
because OBDDs are among the most powerful representations of many NNF subsets in terms of the number of operations supported in polytime, it is surprising that the tractability of operations of tensor trains and OBDDs are identical. 
It is also remarkable that tensor trains emerged from a different field than NNFs; to the best of our knowledge, this is the first study that argues that tensor decomposition provides more succinctness and tractability in representing a Boolean function than NNF subsets.

\section{Related Works}
Since the publication of an innovative paper~\cite{knowledge}, the knowledge compilation map has been extended when new knowledge representation languages appear~\cite{kcmap_extension1,darwiche_2014}.

The knowledge compilation map has contributed to many varied applications and has also been extended beyond Boolean functions~\cite{FargierMN13,choi2020probabilistic}. However, tensor decomposition-based methods have not been evaluated with these maps.
Knowledge compilation maps, which play an important role in understanding knowledge representation languages, have fueled the emergence of several important applications of knowledge compilation techniques~\cite{CHAVIRA2008772,NEURIPS2018_dc5d637e,pmlr-v80-xu18h}.
An approach similar to tensor decomposition is expected to contribute to numerous applications.

Tensor decomposition~\cite{kolda09} is a method to represent a large-scale high-dimensional tensor as the sum of the products of smaller or low-dimensional tensors.
High-dimensional tensors are frequently used in research areas such as machine learning, numerical calculations, and quantum physics. Because an $m$-dimensional tensor has an exponential number of elements with respect to $m$, tensor decomposition methods are indispensable for addressing high-dimensional tensors.
Canonical Polyadic (CP) decomposition~\cite{carroll70} and Tucker decomposition~\cite{tucker66} are the initial examples; other variants have been proposed, including hierarchical Tucker decomposition~\cite{hackbusck09}, which is more succinct than CP decomposition~\cite{khrulkov2018expressive}.
However, these decompositions do not place much emphasis on operations.
As a tensor decomposition that supports several operations, a tensor train~\cite{oseledets2011tensor} was proposed, which decomposes an $m$-dimensional tensor into the product of $m$ 3-dimensional tensors.
Tensor trains were originally invented in the quantum physics community as matrix product states~\cite{fannes92}. Several tensor operations for tensor trains were developed in a previous study whose runtimes are polynomial with the decomposition size\cite{oseledets2011tensor}.

Two other studies bridge Boolean function representations and tensor decompositions: Tensor Decision Diagrams \cite{hong2022tensor} and {Quantum Multiple-valued Decision Diagrams} (QMDD) \cite{niemann2015qmdds}.
Both represent tensors as decision graphs, similar to OBDDs.
However, these studies addressed the representation of real-valued tensors instead of Boolean functions.
Although a method for model counting using tensor decomposition~\cite{dudek2020parallel} exists, this method does not introduce tensor decomposition into the knowledge compilation or analyze its succinctness or tractability.

Our study focuses on the use of tensor decomposition methods as representations of Boolean functions and analyzes their succinctness and tractable Boolean function operations.
To the best of our knowledge, no similar studies have been previously conducted.

\section{Preliminaries}

We briefly introduce the OBDD, tensor train, and knowledge compilation map.
For convenience, we describe the values that a binary variable can assume as $0$ and $1$, which corresponds to $\mathfalse$ and $\mathtrue$ in the ordinary Boolean expression.
Within such an expression, for binary values $a$ and $b$, the conjunction $a\wedge b$, disjunction $a\vee b$, and negation $\neg a$ can be computed by $\min\{a,b\}$, $\max\{a,b\}$, and $1-a$, respectively.

An \emph{$n$-ary Boolean function} considers $n$ binary variables as inputs and returns a binary output.
The set of assignments of the input variables such that the output of the Boolean function is $1$ is known as a \emph{model}.
A Boolean function is \emph{consistent} when it includes a model.
We introduce operations that transform one Boolean function into another.
\begin{dfn}~\cite{darwiche1999compiling}
Let $f$ be a Boolean function and
let $\gamma$ be a consistent term. The conditioning of $f$ on $\gamma$, noted $f |\gamma$, is the Boolean function obtained by assigning each variable $x$ of $f$ by 1 if $x$ is a positive literal of $\gamma$ and by 0 if $x$ is a negative literal of $\gamma$.
\end{dfn}
\begin{dfn}~\cite{knowledge}
Let $f$ be a Boolean function and let $\mathbf{X}$ be a subset of variables.
The \emph{forgetting} of $\mathbf{X}$ from $f$, denoted $\exists \mathbf{X}.f$, is a Boolean function that satisfies $\exists\mathbf{X}.f|x=\exists\mathbf{X}.f|\neg x$
for any $x\in \mathbf{X}$ and $f \models g\Leftrightarrow\exists\mathbf{X}.f \models g$ for any Boolean function $g$ that does not include any variable from $\mathbf{X}$. Here, $f\models g$ implies that $f$ implies $g$; that is, $g(x)=f(x)\wedge g(x)$ for any $x$.
\end{dfn}

\paragraph{OBDD:}
An OBDD is a data structure that represents a Boolean function with a rooted directed acyclic graph, $B=(N, A)$, where $N$ denotes a node set and $A$ denotes a set of oriented edges.
Its structure is as follows. 
We denote the root of the graph by $R \in N$. 
Node set $N$ comprises two terminal $\{\bot, \top\}$ and non-terminal nodes. 
From each non-terminal node $v\in N\setminus \{\bot, \top\}$, exactly two oriented edges exist: \textsf{LO} and \textsf{HI}. We denote the nodes to which these edges indicate as $\textsf{lo}(v)$ and $\textsf{hi}(v)$. 
$v\in N\setminus \{\bot, \top\}$ is given a corresponding label for just one of the input variables, denoted by $\textsf{var}(v)$. 
Each directed edge must have a direction that is consistent with the determined order of the input variables. 
This structure is known as a \emph{variable order}. 
For each node $v\in N$ of the OBDD, the corresponding Boolean function $f_v$ is defined as follows: 
\begin{enumerate}[(i)]
  \item If $v=\bot$, then $f_v=0$ for any assignment of input.
  \item If $v=\top$, then $f_v=1$ for any assignment of input.
  \item If $v\in N\setminus \{\bot, \top\}$, then $f_v=(\lnot \textsf{var}(v)\wedge f_{\textsf{lo}(v)})\vee (\textsf{var}(v)\wedge f_{\textsf{hi}(v)})$.
\end{enumerate}
We denote the function corresponding to the OBDD with R as the root using $f_R$.
The \emph{size} of an OBDD comprises the number of its nodes. 
When considering the operations on two OBDDs, those that can be performed depend on whether their variable orders are identical.
OBDD$_<$ is defined to distinguish between the two.
\begin{dfn}
For the total order $<$ among the input variables, OBDD$_<$ is a subset of the OBDD satisfying the following condition: if node $v_1$ is an ancestor of $v_2$, then $\textsf{\upshape var}(v_1)<\textsf{\upshape var}(v_2)$.
\end{dfn}
Figure~\ref{encode} shows an example of an OBDD with variable order $x_1 < x_2 < x_3$ representing Boolean function $f = (\neg x_1 \wedge \neg x_2) \vee (\neg x_1  \wedge \neg x_3) \vee (\neg x_2 \wedge \neg x_3)$.
Terminal nodes are depicted as squares, whereas internal nodes are represented by circles.
The $\textsf{LO}$ and $\textsf{HI}$ edges are indicated by the dashed and solid arrows, respectively.

\begin{figure*}[t]
  \centering
\includegraphics[width=1.6\columnwidth]{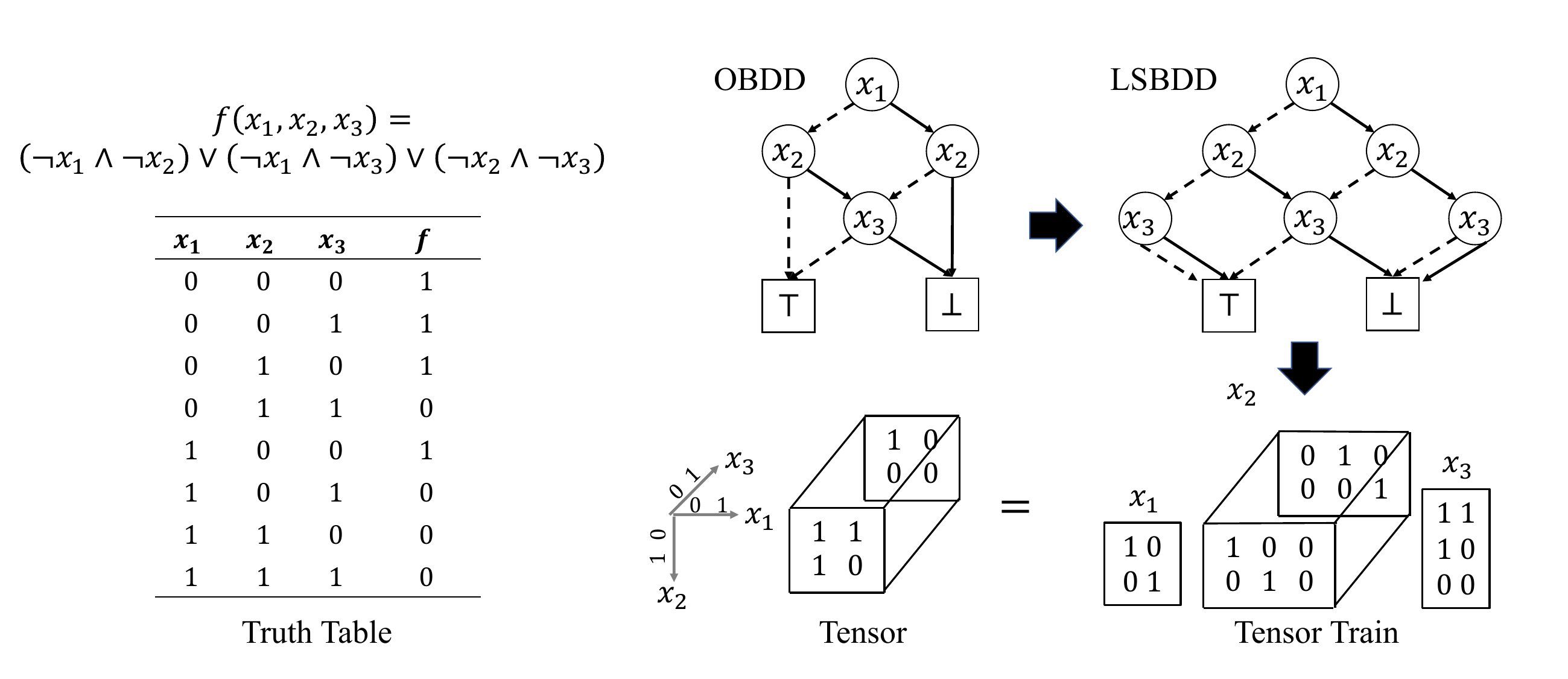} 
  \caption{Example of transforming an OBDD into an equivalent tensor train. We first transform an OBDD into an equivalent level-wise smooth OBDD (LSBDD) and then construct a tensor train from the LSBDD. The OBDD and the LSBDD follow the variable order $x_1 < x_2 < x_3$ and the $\pi$ of the tensor train is defined as $\pi(i) = i$, which follows the order.}
  \label{encode}
\end{figure*}

\paragraph{Tensor train:}
A tensor is an $m$-dimensional array {with $d_1\times\cdots\times d_m$ elements}. 
Each dimension of a tensor is known as a \emph{mode}. 
We denote the element of tensor $\mathcal{X}$ that corresponds to $(a_1, \dots, a_m)$ by $\mathcal{X}(a_1, \dots, a_m)$, where $a_j \in \{0, \ldots, d_j -1\}$ for $1\leq j\leq m$.
For an $m$-dimensional tensor $\mathcal{X}$ of size $d_1\times\cdots\times d_m$ and an $\ell$-dimensional tensor $\mathcal{Y}$ of size $d^\prime_1\times\cdots\times d^\prime_\ell$, where $d_m=d^\prime_1$, \emph{mode product} $\mathcal{X}\mathcal{Y}$ satisfies
\begin{align*}
& \mathcal{X}\mathcal{Y}(a_1,\ldots,a_{m-1},a^\prime_2,\ldots,a^\prime_\ell) \\
& = \sum_{i=0}^{d_m-1}\mathcal{X}(a_1,\ldots,a_{m-1},i)\mathcal{Y}(i,a^\prime_2,\ldots,a^\prime_\ell)
\end{align*}
for any $a_1, \ldots, a_{m-1}, a^\prime_2, \ldots, a^\prime_{\ell}$.
This mode product $\mathcal{X}\mathcal{Y}$ result in an $(m+\ell-2)$-dimensional tensor of size $d_1\times\cdots\times d_{m-1}\times d^\prime_2\times\cdots\times d^\prime_\ell$.
Note that this is a generalization of the matrix product to the tensors.

\emph{Tensor train}~\cite{oseledets2011tensor} represents $d_1\times\cdots\times d_m$ tensor $\mathcal{T}$ using $r_i\times d_i\times r_{i+1}$ tensors $\mathcal{A}_i$ as $\mathcal{T}=\mathcal{A}_1\mathcal{A}_2\cdots \mathcal{A}_m$,
where $r_1=r_{m+1}=1$. 
$r=\max r_i$ is known as the \emph{rank} of the tensor train.
The size of the tensor train representation is the sum of the sizes of $\mathcal{A}_i$.
The size of $\mathcal{A}_i$ is $r_id_ir_{i+1}$, the number of the elements.
Because $r_id_ir_{i+1}\leq r^2d_i$, the size of the tensor train representation is bounded by $r^2\sum_id_i$.
In the following, we focus on $\overbrace{2\times\cdots\times 2}^m$ tensors and denote them as $2^{\times m}$ tensors.
For such tensors, $\sum_id_i=\order{m}$; thus, the number of elements in the tensor train representation is at most $\order{mr^2}$.

In the following, we explain the use of tensor trains to represent Boolean functions.
The following definition implies that a Boolean function can be represented by a tensor while skipping certain variables.

\begin{dfn}
\label{dfn:boolean-function-tensor}
Let $\mathcal{T}$ be a $2^{\times m}$ binary tensor with $1 \leq m \leq n$ and $\pi:[m]  \to [n]$ be a mapping, where $[n] = \{1, \ldots, n\}$. A pair $(\mathcal{T}, \pi)$ is a tensor representation of an $n$-ary Boolean function representing an $n$-ary Boolean function $f$ whose output equals $\mathcal{T}(a_1, \ldots, a_m)$ for the input assignments of variables $x_1, \ldots, x_n$ satisfying $x_j = a_{\pi^{-1}(j)}$ if $j \in Q$, where  $(a_1, \ldots, a_m) \in \{0, 1\}^m$ and $Q = \{\pi(j) \mid j \in [m]\}$.
\end{dfn}

For instance, $f(x_1, x_2) =  x_1\vee \neg x_2$ can be represented by pair $(\mathcal{T}, \pi)$, where $\mathcal{T}$ is $2\times 2$ tensor (= matrix)
\begin{align}
\label{eq:mat-example}
\mathcal{T} = 
\begin{pmatrix}
  1 & 0 \\
  1 & 1 \\
  \end{pmatrix}, 
  \end{align}
  and $\pi(1) = 1$ and $\pi(2) = 2$.
Similarly, $f(x_1, x_2) = \neg x_2$ can be represented as a pair of $\mathcal{T} = (1, 0)^\top$ and $\pi(1) = 2$.

Now, we can assume that a tensor train is a succinct representation of a Boolean function:
\begin{dfn}
Let $(\mathcal{T}, \pi)$ be a tensor representation of an $n$-ary Boolean function $f$, where $\mathcal{T}$ is a $2^{\times m}$ binary tensor. A pair of tensor sequences, $\mathcal{A}_1, \ldots, \mathcal{A}_m$ and $\pi$, is a tensor train representation of $f$, where $\mathcal{A}_i$ for $i \in \{1, \ldots, m\}$ are 3-dimensional ternary tensors; that is, every element of $\mathcal{A}_i$ is in $\{-1,0,1\}$ and satisfies $\mathcal{T} = \mathcal{A}_1 \cdots \mathcal{A}_m$, i.e.,  $\mathcal{A}_1 \cdots \mathcal{A}_m$ is a tensor train decomposition of $\mathcal{T}$.
\end{dfn}
For instance, binary tensor $\mathcal{T}$ in Eq.~\eqref{eq:mat-example} can be represented as a tensor train $\mathcal{T} = \mathcal{A}_1 \mathcal{A}_2$, where $\mathcal{A}_1$ and $\mathcal{A}_2$ are $1 \times 2 \times 2$ and $2 \times 2 \times 1$ tensors defined as:
\begin{align*}
\mathcal{A}_1 = \begin{pmatrix}
  1 & 0 \\
  0 & 1 \\
  \end{pmatrix},~~~
  \mathcal{A}_2 = \begin{pmatrix}
  1 & 0 \\
  1 & 1 \\
  \end{pmatrix}.
\end{align*}

The requirement that $\mathcal{A}_1, \ldots, \mathcal{A}_m$ are ternary tensors ensures that a constant number of bits are required to represent each element of $\mathcal{A}_i$. 
We use this fact to compare the succinctness of the OBDD$_<$ and tensor trains in later sections. 
As shown below, this requirement does not change the set of polytime operations supported by tensor trains.
We use notation $\mathcal{A}_i(\cdot, b, \cdot)$ for representing the slice of the $r \times d \times q$ tensor with $b \in \{0, \ldots, d-1\}$, which corresponds to the $r\times q$ matrix.
Given tensor train $\mathcal{T} = \mathcal{A}_1 \cdots \mathcal{A}_m$, we can represent $\mathcal{T}(a_1, \ldots, a_m)$ using these matrices as follows:
\begin{align}
\label{eq:tensor-train-products}
\mathcal{T}(a_1, \ldots, a_m) = \mathcal{A}_1(\cdot, a_1, \cdot) \cdots \mathcal{A}_m(\cdot, a_m, \cdot)\,.
\end{align}

We say a Boolean function is represented in TT language if it is represented as a pair of tensor trains $\mathcal{A}_1, \ldots, \mathcal{A}_m$ and mapping $\pi$ defined above. Moreover, given a total order $<$ over the input variables $x_1,\ldots,x_n$ of an $n$-ary Boolean function $f$, we define a subclass of TT following $<$.
\begin{dfn}

For a total order $<$ among input variables $x_1, \ldots, x_n$, TT$_<$ is a subclass of TT where mapping $\pi$ satisfies $x_{\pi(i)} < x_{\pi(j)}$ for all $1 \leq i < j \leq m$. We say such mapping $\pi$ {\em follows} variable order $<$.
\end{dfn}

\paragraph{Knowledge compilation map:}
The knowledge compilation map~\cite{knowledge} is an attempt to compare different target representations of knowledge compilation. Since the publication of a seminal paper ~\cite{knowledge}, researchers have been evaluating representations of Boolean functions based on two key dimensions introduced in the paper: the succinctness of the target compilation language and the class of queries and transformations supported by representations in polytime.

Succinctness evaluates the space efficiency of the representations. Because the existing representations for Boolean functions in the literature are all NNF subsets, succinctness has been evaluated by the number of edges in an NNF. 
This assumption is justified by considering the space required for both representations.
The space required for tensor train representation is $\mathcal{O}(\sum_{i=1}^{m}|\mathcal{A}_i|+m\log n)$ bits 
if we limit the possible values of elements of $\mathcal{A}_i$ to $\{-1, 0, 1\}$. Here, $m\leq\sum_{i=1}^{m}|\mathcal{A}_i|$ holds because $|\mathcal{A}_i|>0$ for every $i$.
In contrast, OBDD $B$ needs at least $\Omega(|B| \log{n})$ bits to represent a non-terminal node because every node is associated with a variable. 
Therefore, to retain the tensor train representation of size $S(\coloneqq\sum_{i}|\mathcal{A}_i|)$, we only need the most constant factor of space required for the OBDD representation of the same size $S(\coloneqq |B|)$.

Queries and transformations are operations related to Boolean functions. Although queries do not change the Boolean functions, a transformation returns a modified Boolean function as its output. In knowledge compilation literature, representations are evaluated based on whether they support polytime operations over a set of queries and transformations.
We use the following lemma to determine whether a tensor train supports an operation in polytime:
\begin{lem}
\label{lem:tt-polytime-criteria}
The size of the tensor train representation, $\sum_{i=1}^{m}|\mathcal{A}_i|$, of a Boolean function can be lower bounded by $r+m$, where $r$ denotes the rank and $m$ denotes the number of modes.
\end{lem}
\begin{proof}
Because $r=\max r_i$, at least one tensor $\mathcal{A}_i$ whose size $2r_ir_{i+1}$ is larger than $r$ exists and the size of every tensor is $|\mathcal{A}_i| > 0$.
Therefore,
$\sum_{i=1}^{m}|\mathcal{A}_i| > r + (m - 1)$ and, thus, $\sum_{i=1}^{m}|\mathcal{A}_i|\geq r+m$.
\end{proof}
According to Lemma~\ref{lem:tt-polytime-criteria}, any polynomial comprising $r$ and $m$ is also a polynomial of the size of the tensor train representation.
Therefore, we prove the polytime complexity of operations on tensor train representations by showing the complexities of the polynomials of $r$ and $m$.

\section{Main Results}

We present the results of evaluating the succinctness and tractability of the operations of tensor trains.
First, we introduce the notion of succinctness to compare the expressive efficiencies of the OBDDs and tensor trains.

\begin{dfn}\cite{knowledge}
For two representation classes ($A$ and $B$) of a Boolean function, $A$ is described as at least as \emph{succinct} as $B$, denoted as $A\leq B$, if polynomial $p$ exists such that, for every representation $\beta\in B$, equivalent representation $\alpha\in A$ where $|\alpha|\leq p(|\beta|)$ exists. 
Here, $|\alpha|$ and $|\beta|$ are the sizes of $\alpha$ and $\beta$, respectively.
Furthermore, representation class $A$ is more \emph{succinct} than $B$ if
it is at least as succinct as $B$, even though $B$ is not as succinct as $A$.
\end{dfn}

The following are the main results of this study. 
\begin{thm}
TT$_<$ is more succinct than OBDD$_<$.
\label{thm_succinct}
\end{thm}
This theorem implies that any OBDD $\beta$ of size $|\beta|$ can be translated into tensor train $\alpha$ whose size $|\alpha|$ satisfies $|\alpha| \leq p(|\beta|)$ with polynomial $p$. In contrast, a sequence of Boolean functions $f_1, f_2, \ldots$ exists such that the size of the OBDD representing $f_i$ is exponentially larger than that of the corresponding TT$_<$.

Moreover, we show that all queries and transformations introduced by \cite{knowledge} that can be performed in polytime with OBDD$_<$ can also be performed in polytime with TT$_<$.
Table \ref{table:operation} summarizes an explanation of the computational complexity of each operation.

\begin{table*}
\centering
{\footnotesize
\begin{tabular}{@{}lllc@{}} \toprule
& Name & Explanation & Complexity\\
\midrule
\multirow{8}{*}{\rotatebox[origin=c]{90}{Query}} & CO & Answering whether $f$ has a model. & $\mathcal{O}(mr^2)$ \\
& VA & Answering whether $f$ is valid; that is, $f$ evaluates to $1$ for any assignment of variables. & $\mathcal{O}(mr^2)$ \\
& CE & Answering whether $f \models \gamma$ holds for a clause $\gamma$. & $\mathcal{O}((m+\ell)r^2)$\\
& EQ & Answering whether $f$ and $g$ are equivalent. & $\mathcal{O}((m+\ell)rq(r+q))$ \\
& SE & Answering whether $f \models g$ holds. & $\mathcal{O}((m+\ell)rq(r+q))$ \\ 
& IM & Answering whether $\gamma \models f$ holds for a term $\gamma$. & $\mathcal{O}((m+\ell)r^2)$ \\
& CT & Counting the number of models of $f$. & $\mathcal{O}(mr^2)$ \\
& ME & Enumerating the models of $f$. & $\mathcal{O}(Mmr^2)$ \\
\midrule
\multirow{8}{*}{\rotatebox[origin=c]{90}{Transformation}} &
$\wedge$BC& Constructing TT$_<$ representing the logical conjunction of $f,g$. & $\mathcal{O}((m+\ell)r^2q^2)$ \\
& $\vee$BC & Constructing TT$_<$ representing the logical disjunction of $f,g$. & $\mathcal{O}((m+\ell)r^2q^2)$\\
& $\lnot$C & Constructing TT$_<$ representing the logical negation of $f$. & $\mathcal{O}(mr^2)$ \\
& CD & Constructing TT$_<$ representing $f|\gamma$, 
$f$ conditioned by a $\ell$-literals term $\gamma$. & $\mathcal{O}(\ell r^2)$ \\
& SFO & Constructing TT$_<$ representing $\exists {x}.f$. for a variable $x$. & $\mathcal{O}(mr^4)$ \\
& $\wedge$C & Constructing TT$_<$ representing the logical conjunction of $f_1,f_2,\dots,f_k$. & $\bullet$ \\
& $\vee$C & Constructing TT$_<$ representing the logical disjunction of $f_1,f_2,\dots,f_k$. & $\bullet$ \\
& FO & Constructing TT$_<$ representing $\exists \mathbf{X}.f$ for a variable set $\mathbf{X}$. & $\bullet$ \\
\bottomrule
\end{tabular}
}
\caption{Explanations of operations and computational complexity on tensor trains. We use $m$ and $r$ (resp. $\ell$ and $q$) to represent the number of modes and the tensor train rank corresponding to $f$ (resp. $g$). We also use $\ell$ as the number of modes of a tensor train representing a term or a clause. $M$ denotes the number of models of $f$. Symbol $\bullet$ implies that the operation cannot be performed in polytime on the tensor train regardless of whether P $\neq$ NP.}
\label{table:operation}
\end{table*}

\begin{thm}
\label{thm_operation}
Suppose that Boolean functions $f,g$ are represented as TT$_<$ with an identical variable order $<$.
Then, CO, VA, EQ, SE, CT, $\wedge$BC, $\vee$BC, $\neg$C, and SFO can be performed in polytime with respect to the size of the tensor train representations of $f,g$.
CE, IM, and CD can be performed in polytime with respect to the tensor representation size of $f$ and the number of literals in the clause (for CE) or the term (for IM and CD).
ME can be performed in polytime with respect to the tensor train representation size of $f$ and the number of models.
\end{thm}

We also proved that the tensor train does not support transformations that OBDD$_<$ does not support.
\begin{thm}
\label{thm_nonpoly_operations}
$\wedge$C, $\vee$C, and FO cannot be performed in polytime in the size of tensor trains.
\end{thm}
The proofs of Theorems~\ref{thm_operation} and \ref{thm_nonpoly_operations} are provided in the Appendix.
This high-level concept is presented in the next section.

These theorems show that TT$_<$ achieves a good trade-off between its succinctness and the richness of the supporting operations, as TT$_<$ is more succinct than OBDD$_<$ while it supports the same set of polytime operations. The only existing representation that achieves the same trade-off as OBDD is SDD, which supports the same set of queries and operations as OBDD and is more succinct than OBDD~\cite{bova2016sdds}.
However, TT is clearly different from SDD, because the latter is a subset of NNF, whereas TT is not. Therefore, TT$_<$ is a unique representation that is not equivalent to any existing representation in the knowledge compilation map.

\section{Proofs}
\subsection{High-Level Ideas}

Before proceeding to detailed proofs of the theorems, we describe the high-level ideas.
First, we refer to the succinctness of the tensor train representation (Theorem~\ref{thm_succinct}).
We show that TT$_<$ is as succinct as OBDD$_<$ by providing a transformation scheme from an OBDD to a tensor train in a layer-wise manner; that is, we design each tensor $\mathcal{A}_i$ such that it represents a set of OBDD nodes corresponding to the $i$-th variable.
The transformation is illustrated in Figure~\ref{encode}.
This transformation only increases the size by a polynomial factor and we can prove that TT$_<$ is as succinct as OBDD$_<$.
Second, we show that OBDD$_<$ is not as succinct as TT$_<$ using hidden-weighted bit (HWB) functions, which are known to result in an exponential-size OBDD.
We demonstrate that a rank $2n$ tensor train is sufficient to represent an $n$-ary HWB.

In the proof of Theorem~\ref{thm_operation}, we show that most of the operations can be written as a combination of the sum, inner product, and Hadamard product (i.e., element-wise product) of tensors. 
For instance, considering that output of the tensor train is limited to $\{0,1\}$, the model counting and logical conjunction correspond to the inner and Hadamard products.
The inner and Hadamard products can be performed in polytime on the tensor trains and most other operations can be performed by combining them.

Theorem~\ref{thm_nonpoly_operations} is proven through contradictions. Assuming that the operation is possible in polynomial time, it follows that certain functions can be represented by tensor trains of a polynomial size. However, a contradiction can be observed because a known lower bound exists on the rank of the tensor train.
Owing to space limitations, the proofs of Theorems~\ref{thm_operation} and \ref{thm_nonpoly_operations} are provided in the Appendix.

\subsection{Proof of Theorem~\ref{thm_succinct}}
\label{ssec:proof1}

\paragraph{Proof that TT$_<$ is at least as succinct as OBDD$_<$:}

First, we show that TT$_<$ is at least as succinct as OBDD$_<$; that is, any Boolean functions that can be represented in a size $|\beta|$ OBDD$_<$ can also be represented in a size $|\alpha|$ TT$_<$ satisfying $|\alpha| \leq p(|\beta|)$, where $p$ is a polynomial.
We can prove this by showing a transformation from an OBDD into a tensor train representing the same Boolean function, which does not change the variable order $<$. 
The transformation comprises two steps: (i) transforming the OBDD into a level-wise smooth OBDD (LSBDD) and (ii) transforming the LSBDD into an equivalent tensor train. Because both transformations will cause at most polynomial increases in representations, we can show that the tensor train is as succinct as the OBDD. 

We convert an OBDD into the corresponding LSBDD, which we define as follows:
\begin{dfn}
\label{dfn:lsbdd}
Let $V \subseteq \{x_1, \ldots, x_n\}$ be a subset of variables and we assume a total order $<$ over variables $x_1, \ldots, x_n$ exists. An LSBDD with respect to $V$ is an OBDD where every non-terminal node $v$ satisfies (i) $\textsf{\upshape var}(v) \in V$ and (ii) no $x \in V$ exists such that $\textsf{\upshape var}(v) < x < \textsf{\upshape var}(c)$ for every child $c$ unless $\textsf{\upshape var}(v)$ is the largest variable among $V$.
\end{dfn}
Note that LSBDD differs from a complete OBDD~\cite{bollig2012efficient}, in which every path from the root to a sink node has length $n$.
They differ in that the LSBDD allows variable skipping whereas the complete OBDD does not.
LSBDD is introduced because converting OBDD to complete OBDD increases the representation size by a factor of $n$ at worst, which is not a polynomial in the original representation size.
Figure~\ref{encode} shows an example of OBDD and LSBDD with respect to  $V = \{x_1, x_2, x_3\}$, representing the same Boolean function $f = (\neg x_1 \wedge \neg x_2) \vee (\neg x_1 \wedge \neg x_3) \vee (\neg x_2 \wedge \neg x_3)$.

We demonstrate a procedure for transforming an OBDD into an equivalent LSBDD. Without loss of generality, we assume that the OBDD follows variable order $x_1 < x_2 < \cdots < x_n$.
Let $V(B)$ be the set of labels appearing in OBDD $B$ defined as
$V(B) = \{\textsf{var}(v) \mid v \in N \setminus\{\top, \bot\}\}$. We show that we can construct LSBDD respecting $V(B)$ whose size is bounded by a polynomial of the size of OBDD $B$.
\begin{lem}
\label{lem:obdd-to-lsbdd}
For any OBDD $B$, we can construct an LSBDD equivalent to $B$ with respect to $V(B)$ and equivalent to $B$, whose size is bounded by a polynomial of $|B|$ in polytime.
\end{lem}
\begin{proof}
We can convert any OBDD $B$ into an LSBDD with respect to $V(B)$ by repeating the following procedure: Let $e$ be an attributed edge; that is, either the $\textsf{LO}$ or $\textsf{HI}$ edge in $B$ whose source is $v$ and whose target is $c$.
If $c$ is the terminal and $x \in V(B)$ such that $\textsf{var}(v) < x$ exists, we remove $e$ and add new node $u$ with $\textsf{var}(u) = \max V(B)$ and $\textsf{lo}(u) = \textsf{hi}(u) = c$.
Then, we add edge $e^\prime$ whose source is $v$ and target is $u$.
Edge $e^\prime$ has the same attribute with $e$; that is, if $e$ is a $\textsf{LO}$ (resp. $\textsf{HI}$) edge, then $e^\prime$ is also an $\textsf{LO}$ (resp. $\textsf{HI}$) edge.
If $c$ is a non-terminal node and $x \in V(B)$ satisfying $\textsf{var}(v) < x < \textsf{var}(c)$ exists, we remove $e$ and add new node $u$ with $\textsf{var}(u) = \hat{x}$ and $\textsf{lo}(u) = \textsf{hi}(u) = c$, where $\hat{x}$ is the largest $x \in V(B)$ satisfying $\textsf{var}(v) < x < \textsf{var}(c)$.
Then, we add an edge $e^\prime$ whose source is $v$ and the target is $u$.
The edge $e^\prime$ has the same attribute with $e$. We repeat this substitution until no further substitution is possible to obtain an LSBDD.

We can show that the OBDD obtained by the above procedure is an LSBDD with respect to $V(B)$ because every non-terminal node satisfies the conditions in Definition~\ref{dfn:lsbdd}.
In addition, the obtained LSBDD
represents the same Boolean function because the Boolean function corresponding to $c$ equals $u$ after substitution. Furthermore, we can show that the size of the LSBDD is not larger than $|B|(|V(B)|-1)$ because we add at most $|V(B)|-1$ nodes for every non-terminal node in $B$. Therefore, the size of an LSBDD is bounded by $p(|B|)$, where $p$ is a polynomial because $|V(B)|\leq |B|$.
\end{proof}
Next, we describe the procedure for transforming an LSBDD into TT$_<$. Let $L = (N, E)$ be an LSBDD, where $N$ denotes a set of nodes and $E$ denotes a set of attributed edges. Let $V(L)$ be a set of variables appearing in $L$ and $m = |V(L)|$.
We first define mapping $\pi$ by setting $x_{\pi(i)}$, which is the $i$-th smallest variable in $V(B)$ for $1 \leq i \leq m$.
Next, we define $\mathcal{A}_1, \ldots, \mathcal{A}_m$. 
Let $S_i$ ($1 \leq i \leq m$) be a sequence of all non-terminal nodes $v \in N$ satisfying $\textsf{var}(v) = x_{\pi(i)}$. We order and number the nodes in $S_i$ arbitrarily and denote as $v_{i, 0}, v_{i,1}, \ldots, v_{i, r_i-1}$, where $r_i = |S_i|$. For $1 \leq i < m$, we set $\mathcal{A}_i$ as an $r_i \times 2 \times r_{i+1}$ tensor representing the nodes in $S_i$:
\begin{align*}
\mathcal{A}_i(j,0,k)&=
\begin{cases}
1 & (\textsf{lo}(v_{i,j})=v_{i+1,k}) \\
0 & (\textrm{otherwise})
\end{cases}
,
\end{align*}
\begin{align*}
\mathcal{A}_i(j,1,k)&=
\begin{cases}
1 & (\textsf{hi}(v_{i,j})=v_{i+1,k}) \\
0 & (\textrm{otherwise})
\end{cases}
.
\end{align*}
Similarly, we define $\mathcal{A}_m$ as an $r_m \times 2 \times 1$ tensor representing nodes in $S_m$:
\begin{align*}
\mathcal{A}_m(j,0,0)&=
\begin{cases}
1 & (\textsf{lo}(v_{m,j})=\top) \\
0 & (\textsf{lo}(v_{m,j})=\bot)
\end{cases}, \\
\mathcal{A}_m(j,1,0)&=
\begin{cases}
1 & (\textsf{hi}(v_{m,j})=\top) \\
0 & (\textsf{hi}(v_{m,j})=\bot)
\end{cases}
.
\end{align*}
Each tensor $\mathcal{A}_i$ encodes edges in the LSBDD. In other words, if an edge connecting $v_{i,j}$ to $c_{i+1, k}$ exists, the corresponding element of $\mathcal{A}_i$ is set to 1. Tensor $\mathcal{A}_m$ represents edges connecting nodes in $S_{m}$ to the $\top$ terminal node.

\begin{lem}
\label{lem:lsbdd-equiv-tt}
The pair of tensor sequence $\mathcal{A}_1,\ldots,\mathcal{A}_m$ and mapping $\pi$ obtained by the above procedure is a tensor train representation of the Boolean function that the input LSBDD corresponds to.
\end{lem}
The full proof of Lemma~\ref{lem:lsbdd-equiv-tt} is given in the Appendix.
Lemma~\ref{lem:lsbdd-equiv-tt} shows that TT$_{<}$ can be constructed to represent an equivalent Boolean function with OBDD$_<$.
The number of elements in $\mathcal{A}_i$ is $2r_i r_{i+1}$, where $r_i$ denotes the number of non-terminal LSBDD nodes in $S_i$.
Because $|L|=\sum_{i=1}^{m}r_i$, $\sum_{i}|\mathcal{A}_i|$ is bounded by $|L|^2$, TT$_<$ is at least as succinct as OBDD$_<$.
Because efficient compilation methods for the OBDD~\cite{huang2004using,knuth2011art} exist, they can also be used for the compilation of a tensor train by encoding an OBDD to a tensor train.

\paragraph{Proof that OBDD$_<$ is not as succinct as TT$_<$:}

In the following, we demonstrate the existence of a function that can be represented in a polynomial-sized tensor train but not in a polynomial-sized OBDD.
We consider the HWB function~\cite{bryant1991complexity} defined as follows.
\begin{dfn}
Let $\mathbf{x}=\{x_1,x_2,\dots,x_n\}$ be the input variable. 
The Boolean function {\rm HWB}$_n$ is defined as follows:
\begin{align*}
\rm{HWB}_n(\mathbf{x})=
\begin{cases}0 & (\mathbf{x}=\mathbf{0})\\ 
x_{\sum x_i} & (\text{\upshape otherwise})
\end{cases},
\end{align*}
where $\mathbf{0}$ is an all-0 vector.
\end{dfn}
HWB$_n$ is a function that outputs $x_k$ when the number of $1$s in the input is $k$. 
HWB has long been used to measure the succinctness of Boolean function representations. 
In particular, the following characteristic of OBDDs is known. 
\begin{lem}\cite{bddhwbcomplexity}
The OBDD size of HWB$_n$ is $\Omega(2^{0.2n})$.
\end{lem}
To prove that an OBDD is not as succinct as a tensor train, proving the following proposition is sufficient.
\begin{prop}
{\rm HWB}$_n$ can be represented by a tensor train of size $\mathcal{O}(n^2)$.
\label{prop:HWB_TT}
\end{prop}
To prove the above proposition, we introduce a tensor train representing a Boolean function whose rank is $2n$.
We then show that the representation equals HWB$_n$.

For simplicity, we assume that the variable order is $x_1 <\dots < x_n$. 
An extension is easy for more general ordering. 

Let $\mathcal{A}_0$ be a $2n$-dimensional vector such that only the first element is $1$ and the others are $0$. 
Similarly, $\mathcal{A}_{n+1}$ is defined as a $2n$-dimensional vector such that only the $(n+1)$-st element is $1$ and the others are $0$. 
We also define a sequence of $2n\times 2\times 2n$ tensors $\mathcal{A}_i\ (i\in \{1,2,\dots,n\})$ as follows:
\begin{align*}
\mathcal{A}_i(\cdot,0,\cdot) = 
\begin{pmatrix}
I & O \\
O & I \\
\end{pmatrix},\quad
\mathcal{A}_i(\cdot,1,\cdot) = 
\begin{pmatrix}
I_1 & I^\prime_{i} \\
O & I_{n-1} \\
\end{pmatrix},
\end{align*}
where $O$ is an $n\times n$ matrix whose elements are all $0$, $I$ is an $n\times n$ identity matrix, $I_k$ is an $n \times n$ matrix obtained by cyclically right-shifting each row vector of $I$ by $k$, 
and $I^\prime_k$ is the left-right inversion of $I_k$, i.e., ${I^\prime_k}_{i,j}={I_k}_{i,n-1-j}$.
Under this definition, $\mathbf{v}^\top I_k$ is a vector obtained by cyclically right-shifting the elements of $n$-dimensional vector $\mathbf{v}$ by $k$. 
In addition, we define $\tilde{\mathcal{A}}_i\ (1\leq i\leq n)$ as follows:
\begin{align*}
\tilde{\mathcal{A}}_1 = \mathcal{A}_0\mathcal{A}_1,\quad
\tilde{\mathcal{A}}_i = \mathcal{A}_i,\quad
\tilde{\mathcal{A}}_n = \mathcal{A}_n\mathcal{A}_{n+1}.
\end{align*}
Next, we show that $((\tilde{\mathcal{A}}_1, \ldots, \tilde{\mathcal{A}}_n), \pi)$ is a tensor train representing HWB$_n$, where $\pi: [n] \to [n]$ is an identity function for the assumed variable order.
Let $\mathcal{T}=\mathcal{A}_0\mathcal{A}_1\cdots\mathcal{A}_{n+1}$. The output of the tensor train corresponding to input values $(a_1, \ldots, a_n) \in \{0, 1\}^n$ for variables $x_1, \ldots, x_n$ is $\mathcal{A}_0^\top\mathcal{A}_1(\cdot,a_1,\cdot)\cdots\mathcal{A}_n(\cdot,a_n,\cdot)\mathcal{A}_{n+1}$ by following Eq. (\ref{eq:tensor-train-products}). 

Let $\mathcal{V}_i(a_1, \ldots, a_i)$ be a $2n$-dimensional vector defined as $\mathcal{V}_i(a_1, \ldots, a_i) = \mathcal{A}_0^\top\mathcal{A}_1(\cdot,a_1,\cdot) \cdots \mathcal{A}_i(\cdot,a_i,\cdot)$. We show that $\mathcal{V}_i$ retains sufficient information regarding the values assigned to variables $x_1, \ldots, x_i$ for evaluating HWB$_n$. We design the first $n$ elements of $\mathcal{V}_i$ to encode $\sum_{j=1}^{i}a_j$ and the remaining $n$ elements to store values $a_1, \ldots, a_i$ while maintaining the first element as $a_{\sum a_j}$. Therefore, the $(n+1)$-st element of $\mathcal{V}_m$ is equal to the HWB$_n$ output. The following lemma shows that $\mathcal{V}_i$ stores such information.
\begin{lem}
\label{lem:hwb-lemma}
For any $k\in\{0,1,\dots,n\}$ and $(a_1, \ldots, a_k) \in \{0, 1\}^{k}$, $\mathcal{V}_k = \mathcal{A}_0\mathcal{A}_1(\cdot,a_1,\cdot)\cdots\mathcal{A}_k(\cdot,a_k,\cdot)$ is a $2n$-dimensional vector that satisfies the following properties:
(i) For the first $n$ elements, only one element is $1$ and the others are $0$, and the $j$-th element is $1$ iff $\sum_{i=1}^k a_i = j-1\ (\text{\upshape mod}\ n)$.
(ii) The remaining $n$ elements are made by cyclically left-shifting $(a_1,\dots,a_k,0,\dots,0)$ by $j$. 
Therefore, the $(n+1)$-th element of $\mathcal{V}_k$ equals $a_{\sum_{i=1}^{k}a_i}$.
\label{lemma}
\end{lem}
\begin{proof}
This is shown by mathematical induction and holds for $k=0$. 
Assuming this holds for $k=i$, we divide the case according to $a_{i+1}$. 
When $a_{i+1}=0$, $\mathcal{V}_{i+1} = \mathcal{V}_i \mathcal{A}_{i+1}(\cdot, 0, \cdot) = \mathcal{V}_i$.
This satisfies (i) and (ii) because the number of $1$'s is the same as  $\mathcal{V}_i$ and $(a_1, \ldots, a_i, 0, \ldots,  0) = (a_1, \ldots, a_{i+1}, 0, \ldots, 0)$. 
When $a_{i+1}=1$, the first $n$ elements of $\mathcal{V}_{i+1}$ are obtained by circularly right-shifting the first $n$ elements of $\mathcal{V}_i$ by 1 by $\mathcal{V}_{i} \mathcal{A}_{i+1}(\cdot,1,\cdot)$ from the definition of $\mathcal{A}_{i+1}$. 
The remaining $n$ elements of $\mathcal{V}_{i+1}$ are made by cyclically left-shifting $(a_1, \ldots, a_{i+1}, 0, \ldots, 0)$ 
such that the first element corresponds to $a_{\sum a_i}$.
This vector is obtained from the remaining $n$ elements of $\mathcal{V}_{i}$ by setting the element corresponding to $x_{i+1}$ to 1 and circularly left-shifting the vector by 1.
This is also performed by $\mathcal{V}_i\mathcal{A}_{i+1}(\cdot, 1, \cdot)$.
These results demonstrate the validity of the subject through mathematical induction. 
\end{proof}
\begin{proof}[Proof of Proposition~\ref{prop:HWB_TT}]
By Lemma~\ref{lem:hwb-lemma}, the $(n+1)$-st element of $\mathcal{V}_n = \mathcal{A}_0\mathcal{A}_1(\cdot,a_1,\cdot)\cdots\mathcal{A}_n(\cdot,a_n,\cdot)$ corresponds to the $(\sum_{i=1}^n a_i)$-th value of the input; that is, it corresponds to HWB$_i(\mathbf{x})$. 
Multiplying $\mathcal{A}_{n+1}$ from the right corresponds to selecting the $(n+1)$-st element of $\mathcal{V}_n$, the tensor train representing HWB$_n$.
Therefore, $((\tilde{\mathcal{A}_1},\dots,\tilde{\mathcal{A}_n}), \pi)$ is a tensor train representing HWB$_n$ with $\mathcal{O}(n^2)$ size.
\end{proof}

Note that the above proofs hold even if the elements of $\mathcal{A}_i$ are limited to $\{-1,0,1\}$.

\subsection{Proof Sketch of Theorem~\ref{thm_operation}}

The proof of Theorem~\ref{thm_operation} is rather long, so we will only explain the proof sketch.
For details, please refer to the Appendix.

We can verify that, within the tensor representation in Definition~\ref{dfn:boolean-function-tensor}, every operation in Table~\ref{table:operation} can be reduced to basic tensor operations, such as the sum, Hadamard product, and inner product defined below.
\begin{dfn}
For two $d_1\times\dots\times d_m$ tensors $\mathcal{A},\mathcal{B}$, their sum $+$, Hadamard product $\circ$, and inner product $\cdot$ are defined as follows:
\begin{align*}(\mathcal{A}\!+\!\mathcal{B})(a_1,\dots,a_m) &= \mathcal{A}(a_1,\dots,a_m) + \mathcal{B}(a_1,\dots,a_m), \\
(\mathcal{A}\!\circ\!\mathcal{B})(a_1,\dots,a_m) &= \mathcal{A}(a_1,\dots,a_m) \mathcal{B}(a_1,\dots,a_m), \\
\mathcal{A}\cdot\mathcal{B} &= \smashoperator{\sum_{a_1,\dots,a_m}}\ \mathcal{A}(a_1,\dots,a_m) \mathcal{B}(a_1,\dots,a_m), 
\end{align*}
where $\mathcal{A}+\mathcal{B}$ and $\mathcal{A}\circ\mathcal{B}$ are $d_1\times\dots\times d_m$ tensors and $\mathcal{A}\cdot\mathcal{B}$ is a scalar value.
\end{dfn}

For example, we show that CT, the query of counting the number of models, is reduced to taking inner product with an all-one tensor and $\wedge$BC, the transformation of taking the logical conjunction of two Boolean function, is reduced to taking Hadamard product of two tensors.
It is also known that these operations can be performed efficiently even if the tensors are represented as tensor trains of the same mode order~\cite{oseledets2011tensor}.

However, to perform tensor operations, the modes of the tensors must be aligned.
Within tensor train representations, this implies that we should align their mappings $\pi$.
Therefore, we prepare the method used to align the mappings.
In addition, we show that clauses and terms can be represented by a constant-rank tensor train.
Finally, by aligning the mappings of two tensor train representations, we show that the operations listed in Table~\ref{table:operation}, except for $\wedge$C, $\vee$C, and FO, can be reduced to the combinations of sum, Hadamard product, and inner product.

\subsection{Proof of Theorem~\ref{thm_nonpoly_operations}}
In this section, we prove that $\wedge$C, $\vee$C, and FO cannot be performed in polynomial time on tensor trains.

\subsubsection{$\wedge$C, $\vee$C:}
Assume that $\wedge$C can be performed in polynomial time on a tensor train.
Let $f=(x_1\vee \neg y_1)\wedge(\neg x_1\vee y_1)\wedge(x_2\vee \neg y_2)\wedge(\neg x_2\vee y_2)\wedge\dots\wedge(x_n\vee \neg y_n)\wedge(\neg x_n\vee y_n)$ and consider representing $f$ by a tensor train with order $x_1<\dots <x_n<y_1<\dots<y_n$.
Because each clause can be represented by a linear size tensor train, $f$ can also be represented by a polynomial size (polysize) tensor train based on this assumption.

Let $A(x_1,\dots, x_n;y_1,\dots,y_n)$ be the unfolding matrix of the tensor representation of $f$; that is, $(x_1, \dots, x_n)$ corresponds to the row, and $(y_1, \dots, y_n)$ corresponds to the column.
The rank of this matrix is the lower bound of the tensor train rank~\cite{oseledets2011tensor}.
Because $f$ outputs $1$ if and only if $(x_1,\dots,x_n)=(y_1,\dots,y_n)$, $A(x_1,\dots, x_n;y_1,\dots,y_n)$ is an identity matrix; thus, its rank is $2^n$.
Therefore, the rank of the tensor train is also $\Omega(2^n)$, contradicting that $f$ can be represented by a polynomial size tensor train.

By considering $\neg f$, we can similarly prove $\vee$C with the power of De Morgan's law.

\subsubsection{FO:}
Assume that FO can be performed in polynomial time on tensor trains.
Let $f=(x_1\vee \neg y_1)\wedge(\neg x_1\vee y_1)\wedge(x_2\vee \neg y_2)\wedge(\neg x_2\vee y_2)\wedge\dots\wedge(x_n\vee \neg y_n)\wedge(\neg x_n\vee y_n)$.
Let $C_i$ be the $i$-th clause of $f$ and $h=(z_1\wedge \neg C_1)\vee(\neg z_1\wedge z_2\wedge \neg C_2)\vee\dots(\neg z_1\wedge\dots\neg z_{2n-1} \wedge z_{2n}\wedge\neg C_{2n})$.
$h$ can be represented by a polysize OBDD with order $z_1<\dots <z_{2n}<x_1<\dots<x_n< y_1<\dots<y_n$, which has nodes $v_1,\dots,v_{2n}$ whose labels are $\textsf{var}(v_i)=z_i$ for $1\leq i\leq 2n$, $\textsf{lo}(v_i)=v_{i+1}$ for $1\leq i\leq 2n-1$, $\textsf{lo}(v_{2n})=\bot$.
The substructure below $\textsf{hi}(v_{i})$ represents $C_i$ for $1\leq i\leq 2n$.
As every $C_i$ can be represented by a polysize OBDD, the overall OBDD representing $h$ is also polysized.
According to Theorem~\ref{thm_succinct}, $h$ can be represented by a polysize tensor train.
In addition, by the definition of forgetting, we have $\exists\{z_1,\ldots,z_{2n}\}.h=\neg C_1\vee\dots\vee\neg C_{2n}=\neg f$.
By assumption, we obtain a polysize tensor train representing $\neg f=\neg C_1\vee\dots\vee\neg C_{2n}.$
Because $\neg$C is also tractable on tensor trains, we obtain a polysize tensor train representing $f$.
However, this contradicts the lower bound of the rank of $f$ mentioned in the above proof regarding $\wedge$C.
Therefore, we cannot perform FO in polynomial time on tensor trains.

\section{Discussion}
Our theoretical results demonstrate that the power of the tensor trains is equivalent to that of OBDDs. However, their limitations persist as representations of Boolean functions. For instance, a bounded conjunction is performed using the Hadamard product of tensor trains and the rank of the output tensor always becomes the product of the ranks of the input tensor trains, whereas the output OBDDs tend to be small in practice. Future work on TT-based representations will include finding a way to perform practically efficient operations.

In this study, we compared the succinctness of TT using OBDD.
It is difficult to find succinctness between the TT and the other representations if we apply standard strategies. 
It is evident that no encoding from the NNF subsets exists that does not support the same set of queries (e.g., DNNF) to TT. 
SDD~\cite{darwiche2011sdd} is the only representation supporting the same set of polytime operations and it is expected that TT will not be more succinct than SDD because it has a more flexible structure for variables than OBDD and TT. 
Conversely, the NNF subsets are not expected to be more succinct than TT because they cannot represent +1,-1 weights. 

Raw tensor representation supports parallelism. Similarly, given that a tensor train comprises multiple tensors, we can perform calculations in parallel, as demonstrated in \cite{yang2017tensor}, thereby benefiting from GPU acceleration. One advantage of tensor-based parallel processing is its independence from variable linear ordering, which eliminates the need to consider order.

\section{Conclusion}
This study examined the succinctness and richness of the operations of the tensor train as a representation of Boolean functions by comparing them with the OBDD. 
We proved that the tensor train is more succinct than the OBDD and can efficiently perform every operation that the OBDD can perform in polytime. These results suggest that
the tensor decomposition is a powerful representation of Boolean functions.
Because encoding from OBDD to TT exists, compilation for OBDD can also be used for TT.
However, it is yet to be determined whether a better compilation method exists.
The relationship between tensor trains and SDD or other knowledge representations is a topic for future research.

\section{Acknowledgements}
The authors thank the anonymous reviewers for their valuable feedback, corrections, and suggestions. This work was supported by JST CREST (Grant Number JPMJCR22D3, Japan) and JSPS KAKENHI (Grant Number JP20H05963, Japan).

\bibliography{aaai25}

\appendix
\section{Appendix: Full Proofs}
In this appendix, we show missing proofs.
First, we show the proof of Lemma~\ref{lem:lsbdd-equiv-tt} that is used in the proof of Theorem~\ref{thm_succinct}.
Then, we show the full proofs of Theorems~\ref{thm_operation} and \ref{thm_nonpoly_operations}.

\subsection{Proof of Lemma~\ref{lem:lsbdd-equiv-tt}}
\begin{proof}
We use induction to prove the equivalence between LSBDD and TT.
First, we prove the base step: we show that the node $v_{m,j}$ in LSBDD and the $j$-th vector $(\mathcal{A}_m(j, 0, 0), \mathcal{A}_m(j, 1, 0))^\top$ in $\mathcal{A}_m$ represents the same Boolean function.
A total of four possible pairs of \textsf{LO} and \textsf{HI} child nodes for non-terminal nodes in $S_m$ exist: $(\top, \top)$, $(\top, \bot)$, $(\bot, \top)$, and $(\bot, \bot)$, each of which corresponds to Boolean functions $1, \neg x_{\pi(m)}, x_{\pi(m)}$, and $0$, where $1$ and $0$ represent identity functions evaluated as $1$ and $0$, respectively.
According to the definition of $\mathcal{A}_m$, the vectors corresponding to four different non-terminal nodes $(\top, \top)$, $(\top, \bot)$, $(\bot, \top)$, and $(\bot, \bot)$ are $(1, 1)^\top, (1, 0)^\top, (0, 1)^\top$, and $(0, 0)^\top$.
Here, they represent Boolean functions $1, \neg x_{\pi(m)}, x_{\pi(m)}$, and $0$.
Therefore, the base step holds.

Next, we show the inductive step.
Let $\mathcal{A}_{i, j}$ be a $1 \times 2 \times r_{i+1}$ tensor satisfying $\mathcal{A}_{i, j}(0, b, k) = \mathcal{A}_{i}(j, b, k)$ for every $b \in \{0, 1\}$ and $0 \leq k < r_{i+1}$, where $r_{m+1} = 1$ corresponding to the definition of tensor train.
As an inductive step, we assume that, for any $0\leq j< r_{i+1}$, the pair of $\mathcal{A}_{i+1,j},\mathcal{A}_{i+2},\ldots,\mathcal{A}_{m}$ and $\pi$ is a tensor train representation of the Boolean function that the LSBDD node $v_{i+1, j}$ represents.
We then show that, for any $0\leq j<r_i$, the pair of $\mathcal{A}_{i, j}, \mathcal{A}_{i+1},\ldots,\mathcal{A}_m$ and $\pi$ is a tensor train representation of the Boolean function that $v_{i, j}$ represents.
Let $\mathcal{T}_{i, j} = \mathcal{A}_{i, j} \mathcal{A}_{i+1}\cdots \mathcal{A}_m$.
From the assumption, every $\mathcal{T}_{i+1, j}$ represents a Boolean function corresponding to $v_{i+1, j}$.
Then, from the definition of $\mathcal{A}_i$, $\mathcal{T}_{i,j}$ can be described as follows:
\begin{align*}
\mathcal{T}_{i, j}(0, a_{i+1}, \ldots, a_{m})&= \mathcal{T}_{i+1, l}(a_{i+1}, \ldots, a_m), \\
\mathcal{T}_{i, j}(1, a_{i+1}, \ldots, a_{m})&= \mathcal{T}_{i+1, h}(a_{i+1}, \ldots, a_m) ,
\end{align*}
where $l$ (resp. $h$) satisfies $v_{i+1, l} = \textsf{lo}(v_{i, j})$ (resp. $v_{i+1, h} = \textsf{hi}(v_{i, j})$. Therefore, $\mathcal{T}_{i, j}$ is the tensor representation of the Boolean function equivalent to $v_{i, j}$.

By the principle of mathematical induction, the pair of $\mathcal{A}_1,\ldots,\mathcal{A}_m$ and $\pi$ is a tensor train representation of the Boolean function that the root node of LSBDD represents.
This proves Lemma~\ref{lem:lsbdd-equiv-tt}.
\end{proof}

\subsection{Proof of Theorem~\ref{thm_operation}}

\begin{lem}
\label{lem:expand-tt}
    Let $((\mathcal{A}_1, \ldots, \mathcal{A}_m), \pi)$ be a tensor train representation of an $n$-ary Boolean function, where $m < n$. Then the pair $((\mathcal{A}_1, \ldots, \mathcal{A}_{k}, \mathcal{I}_k, \mathcal{A}_{k+1}, \ldots, \mathcal{A}_{m}), \pi^\prime)$ of a tensor train and a mapping for $0 \leq k \leq m$ is a tensor train representation corresponding to the same Boolean function, where $\mathcal{I}_k$ is $r_{k+1} \times 2 \times r_{k+1}$ tensor defined as $\mathcal{I}_k(\cdot, 0, \cdot) = \mathcal{I}_k(\cdot, 1, \cdot) = I$ and $\pi^\prime: [m+1] \to [n]$ is an injection satisfying $\pi^\prime(i) = \pi(i)$ if $i < k$ and $\pi^\prime(i) = \pi(i+1)$ if $i > k$.
\end{lem}
\begin{proof}
Let $\mathcal{T} = \mathcal{A}_1 \cdots \mathcal{A}_m$ and $\mathcal{T}^\prime = \mathcal{A}_1 \cdots \mathcal{I}_k \cdots \mathcal{A}_{m}$. 
Because $\mathcal{A}_k(\cdot, a_k, \cdot) \mathcal{I}_k(\cdot, b, \cdot) \mathcal{A}_{k+1}(\cdot, a_{k+1}, \cdot)$ equals $\mathcal{A}_k(\cdot, a_k, \cdot)  \mathcal{A}_{k+1}(\cdot, a_{k+1}, \cdot)$ for any $(a_{k}, a_{k+1}, b) \in \{0, 1\}^3$, $\mathcal{T}(a_1, \ldots, a_m) = \mathcal{T}^\prime(a_1, \ldots, a_k, b, a_{k+1}, \ldots, a_m)$ for any $(a_1, \ldots, a_m) \in \{0, 1\}^m$ and $b\in\{0,1\}$.
Thus, $(\mathcal{T}^\prime, \pi^\prime)$ represents the same Boolean function with $(\mathcal{T}, \pi)$.
\end{proof}
Using the above lemma, we show that the two tensor trains $(\mathcal{A}, \pi)$ and $(\mathcal{B}, \rho)$ following the same variable order can be translated into equivalent tensor trains $(\mathcal{A}^\prime, \tau)$ and $(\mathcal{B}^\prime, \tau)$ having the same number of modes and identical mapping $\tau$.
\begin{lem}
\label{lem:align-tt}
Let $((\mathcal{A}_1, \ldots, \mathcal{A}_m), \pi)$ and $((\mathcal{B}_1, \ldots, \mathcal{B}_\ell), \rho)$ be tensor train representations of $n$-ary Boolean functions $f$ and $g$, where $\pi: [m] \to [n]$ and $\rho: [\ell] \to [n]$ are injections that follow the same variable order. Then, tensor train representations $((\mathcal{A}^\prime_1, \ldots, \mathcal{A}^\prime_k), \tau)$ and $((\mathcal{B}^\prime_1, \ldots, \mathcal{B}^\prime_k), \tau)$ of $f$ and $g$ exist. Moreover, the size of $\mathcal{A}^\prime_1, \ldots, \mathcal{A}^\prime_k$ and $\mathcal{B}^\prime_1, \ldots, \mathcal{B}^\prime_k$ are bounded by a polynomial of the sizes of $\mathcal{A}_1, \ldots, \mathcal{A}_m$ and $\mathcal{B}_1, \ldots, \mathcal{B}_\ell$.
The time complexity is $\mathcal{O}((m+\ell)(r^2+q^2))$.
\end{lem}
\begin{proof}
For simplicity, we assume that the tensor trains follow variable order $x_1 < \cdots < x_n$.
We show a procedure that transforms $((\mathcal{A}_1, \ldots, \mathcal{A}_m), \pi)$ into $((\mathcal{A}^\prime_1, \ldots, \mathcal{A}^\prime_k), \tau)$.
Let $j \in [n]$ be an index satisfying $\exists i: \rho(i) = j$ and $\forall i: \pi(i) \neq j$.
By following Lemma~\ref{lem:expand-tt}, we extend the tensor train representation by inserting $\mathcal{I}_w$ at position $w$, where $w$ is the largest integer satisfying $\pi(w) < j$.
We then update $\pi$ to $\pi^\prime$ following the procedure in Lemma~\ref{lem:expand-tt}, while setting $\pi^\prime(w) = j$.
We repeat the procedure every $j$ to obtain
$((\mathcal{A}^\prime_1, \ldots, \mathcal{A}^\prime_k), \tau)$.
A similar process yields $((\mathcal{B}^\prime_1, \ldots, \mathcal{B}^\prime_k), \tau)$.

We show that the size of $\mathcal{A}^\prime_1, \ldots, \mathcal{A}^\prime_k$ is bounded by a polynomial of the sizes of $\mathcal{A}_1, \ldots, \mathcal{A}_m$ and $\mathcal{B}_1, \ldots, \mathcal{B}_\ell$. Because the output is obtained by inserting tensors $\mathcal{I}$ at most $\ell$ times, the increase in the size of the output tensor is bounded by $\ell r^2$, where $r$ is the rank of $\mathcal{A}_1,\ldots,\mathcal{A}_m$. Similarly, the size of $\mathcal{B}^\prime_1, \ldots, \mathcal{B}^\prime_k$ is bounded by the size of $\mathcal{B}_1,\ldots,\mathcal{B}_\ell$ plus $mq^2$, where $q$ is the rank of $\mathcal{B}_1,\ldots,\mathcal{B}_\ell$.
The time complexity is $\mathcal{O}((m+l)(r^2+q^2))$ because $k\leq m+l$ and the number of elements of output tensor trains is $\mathcal{O}(k(r^2+q^2))$.
\end{proof}

Note that, after the alignment, the number of modes $k$ is at most $m+\ell$.
Moreover, the ranks of the tensor train representations remain the same.

For two tensor train representations $((\mathcal{A}_1,\ldots,\mathcal{A}_m),\pi)$ and $((\mathcal{B}_1,\ldots,\mathcal{B}_\ell),\rho)$, we can consider their sum, Hadamard product, and inner product after aligning their mappings by Lemma~\ref{lem:align-tt}.
Such operations can be performed efficiently within tensor train representations.
\begin{lem}
  \label{lem:basic-operations-tt}
  Let $m,\ell$ be the number of modes of $\mathcal{A}=\mathcal{A}_1\cdots\mathcal{A}_m$ and $\mathcal{B}=\mathcal{B}_1\cdots\mathcal{B}_\ell$, and $r,q$ be the tensor train ranks of $\mathcal{A}$ and $\mathcal{B}$.
  Then, after the alignment of mappings by Lemma~\ref{lem:align-tt} to obtain $\mathcal{A}^\prime=\mathcal{A}^\prime_1\cdots\mathcal{A}^\prime_k$ and $\mathcal{B}^\prime=\mathcal{B}^\prime_1\cdots\mathcal{B}^\prime_k$, the tensor train representations of their sum and Hadamard product can be computed in $\mathcal{O}((m+\ell)(r+q)^2)$ and $\mathcal{O}((m+\ell)r^2q^2)$ times.
  Their inner product can be computed in $\mathcal{O}((m+\ell)rq(r+q))$ time.
\end{lem}
\begin{proof}
  Following~\cite{oseledets2011tensor}, the sum $\mathcal{C}$ of $\mathcal{A}^\prime$ and $\mathcal{B}^\prime$ can be represented as $\mathcal{C}=\mathcal{C}_1\cdots\mathcal{C}_k$, where $\mathcal{C}_1=\begin{bmatrix}\mathcal{A}^\prime_1 & \mathcal{B}^\prime_1\end{bmatrix}$, $\mathcal{C}_k=\begin{bmatrix}\mathcal{A}^\prime_1 \\ \mathcal{B}^\prime_1\end{bmatrix}$ and $\mathcal{C}_i(\cdot,j,\cdot)=\begin{bmatrix}\mathcal{A}^\prime_i(\cdot,j,\cdot) & O \\ O & \mathcal{B}^\prime_i(\cdot,j,\cdot)\end{bmatrix}$ for $1< i< k$.
  The tensor train rank of $\mathcal{C}$ is at most $r+q$.
  Thus, the time complexity is bounded by $\mathcal{O}(k(r+q)^2)=\mathcal{O}((m+\ell)(r+q)^2)$.
  The Hadamard product $\mathcal{D}$ of $\mathcal{A}^\prime$ and $\mathcal{B}^\prime$ is $\mathcal{D}=\mathcal{D}_1\cdots\mathcal{D}_k$, where $\mathcal{D}_i(\cdot,j,\cdot)=\mathcal{A}^\prime_i(\cdot,j,\cdot)\otimes \mathcal{B}^\prime_i(\cdot,j,\cdot)$ for $1\leq i\leq k$, which also follows~\cite{oseledets2011tensor}.
  Here, $\otimes$ is the Kronecker product of matrices.
  The tensor train rank of $\mathcal{D}$ is at most $rq$.
  Thus, the complexity is bounded by $\mathcal{O}(k(rq)^2)=\mathcal{O}((m+\ell)r^2q^2)$.

  For the inner product, we have to compute the sum of the elements of Hadamard product $\mathcal{D}$.
  Although $\mathcal{D}$ has $2^k$ elements, 
  the summation can be distributed over every mode; that is, we have $\sum_{a_1,\ldots,a_k}\mathcal{D}(a_1,\cdots,a_k)=(\sum_{a_1}\mathcal{D}_1(\cdot,a_1,\cdot))\cdots(\sum_{a_k}\mathcal{D}_k(\cdot,a_k,\cdot))$.
  Thus, we do not need a summation over several values, yielding the polytime complexity.
  Moreover, following the procedure of~\cite{oseledets2011tensor} to obtain the inner product on tensor trains, we do not need to explicitly compute the individual Kronecker product $\mathcal{D}_i(\cdot,j,\cdot)=\mathcal{A}^\prime_i(\cdot,j,\cdot)\otimes \mathcal{B}^\prime_i(\cdot,j,\cdot)$.
  The complexity is bounded by $\mathcal{O}(krq(r+q))=\mathcal{O}((m+\ell)rq(r+q))$.
\end{proof}

Next, we show that any term or clause can be represented as a tensor train of constant rank.
A term $\gamma=\bigwedge_{j=1}^{\ell}l_j$ is a conjunction of $\ell$ literals, where each literal $l_j$ is either $x_t$ or $\neg x_t$ for some $1\leq t\leq n$ and no two different literals share the same variable.
Similarly, a clause $\gamma^\prime=\bigvee_{j=1}^{\ell}l_j$ is a disjunction of $\ell$ literals.
\begin{lem}
  \label{lem:simple-tt}
  For any order $<$ among variables, we can construct a tensor train representation of a term whose number of modes is $\ell$ and rank is $1$.
  Additionally, we can construct a tensor train representation of a clause whose number of modes is $\ell$ and rank is $2$.
\end{lem}
\begin{proof}
  Let $V(\gamma)$ be the set of variables appearing in the term $\gamma$ and let $\pi:[\ell]\rightarrow[n]$ be an injection where $\pi(i)=j$ if and only if $x_j$ is the $i$-th smallest variable among $V(\gamma)$ according to the order $<$.
  We define $\mathcal{A}_i$ as a $1\times 2\times 1$ tensor that equals $(0, 1)$ when $x_{\pi(i)}$ appears in $\gamma$ and $(1, 0)$ if $\neg x_{\pi(i)}$ appears in $\gamma$.
  Then, $((\mathcal{A}_1,\ldots,\mathcal{A}_\ell),\pi)$ is a rank-1 tensor train representation of $\gamma$.

  For clauses, we construct mapping $\pi$ in the same manner as above.
  According to De Morgan's law, we have $\bigvee_{j=1}^{\ell}l_j=\neg(\bigwedge_{j=1}^{\ell}\neg l_j)$, where $\bigwedge_{j=1}^\ell\neg l_j$ can be represented by rank-1 tensor train $\mathcal{A}$ because it is a term.
  Within tensor representation, negation can be performed by $\mathbf{1} + (-\mathcal{A})$, where $\mathbf{1}$ is a $2^{\times \ell}$ tensor whose elements are all $1$.
  $\mathbf{1}$ can be represented as a rank-1 tensor train $\mathcal{B}_1\cdots \mathcal{B}_\ell$, where $\mathcal{B}_i$ is a $1 \times 2 \times 1$ tensor $(1, 1)$.
  Additionally, to obtain $-\mathcal{A}$ from $\mathcal{A}$, we multiply $\mathcal{A}_1$ by $-1$.
  Finally, by Lemma~\ref{lem:basic-operations-tt}, the sum of $-\mathcal{A}$ and $\mathbf{1}$ is a rank-2 tensor train with $\ell$ modes.
  Note that although the multiplication with $-1$ possibly generates a $(-1)$-valued element, it still satisfies that every element is in $\{-1,0,1\}$.
\end{proof}
Finally, we demonstrate that each operation can be performed in polytime.
In the following, we assume that Boolean functions $f$ and $g$ are represented by the tensor train representations $(\mathcal{A},\pi)$ and $(\mathcal{B},\rho)$, where the ranks of $\mathcal{A}$ and $\mathcal{B}$ are $r$ and $q$, and the number of modes of each tensor is $m$ and $\ell$, respectively.
Note that, for CE, IM, and CD, $\ell$ denotes the number of literals in the term or clause.

\subsubsection{CT:}
CT corresponds to counting the number of 1-elements of $2^{\times m}$ tensor $\mathcal{A}$ and then multiplying it with $2^{n-m}$.
The former operation can be performed by taking the inner product  $\mathbf{1}\cdot\mathcal{A}$, where $\mathbf{1}$ is a $2^{\times m}$ tensor whose elements are all $1$.
As in the proof of Lemma~\ref{lem:simple-tt}, $\mathbf{1}$ can be represented with a rank-1 tensor train with $\ell=m$ modes.
By substituting the complexity of the inner product in Lemma~\ref{lem:basic-operations-tt} with $\ell=m$ and $q=1$, CT can be answered in $\mathcal{O}(mr^2)$ time.

Moreover, we can perform weighted model counting where the weight of a model is represented by the product of 
weights of each variable as the inner product $\mathcal{W} \cdot \mathcal{A}$, where $\mathcal{W} = \mathcal{W}_1 \cdots \mathcal{W}_m$ is a rank-1 real-valued tensor train.
Each $\mathcal{W}_i$ is a $1\times 2 \times 1$ tensor having $(\bar{w}_{\pi(i)}, w_{\pi(i)})$, where $\bar{w}_j, w_j$ are negative and positive weights corresponding to variable $x_j$.
Note that the procedure for obtaining the inner product in Lemma~\ref{lem:basic-operations-tt} is also valid for real-valued tensor trains.
The final count value can be obtained by multiplying it with $\prod_{i:\nexists j: \pi(j)=i}(\bar{w}_{i}+w_{i})$.

\subsubsection{CO, VA, CE, EQ, SE, and IM:}
These operations can be performed using CT: 

\begin{itemize}
\item CO can be determined by whether the result of CT does not equal $0$. 
\item VA can be determined by whether the result of CT equals $2^n$.
\item EQ can be determined as follows:
We first align the mapping by Lemma~\ref{lem:align-tt} to obtain $\mathcal{A}^\prime$ and $\mathcal{B}^\prime$.
It can then be determined by whether $\mathcal{A}^\prime\cdot\mathcal{B}^\prime$, $\mathbf{1}\cdot\mathcal{A}^\prime$, and $\mathbf{1}\cdot\mathcal{B}^\prime$ are all equal.
\item SE can be determined by first aligning the mappings by Lemma~\ref{lem:align-tt} and then checking whether $\mathcal{A}^\prime\cdot\mathcal{B}^\prime$ equals $\mathbf{1}\cdot\mathcal{A}^\prime$.
\end{itemize}
We can also answer CE and IM in the same way as SE because any term and clause can be represented by a constant-rank tensor train with $\ell$ modes by Lemma~\ref{lem:simple-tt}.
The complexities of CO and VA are identical to that of $\mathcal{O}(mr^2)$ in CT.
The complexities of EQ and SE are identical to that of $\mathcal{O}((m+\ell)rq(r+q))$ in the inner product.
The complexity of CE and IM is $\mathcal{O}((m+\ell)r^2)$, where $\ell$ denotes the number of literals in the term or clause.
Note that, before performing EQ, SE, CE and IM, we need to align the variables using Lemma~\ref{lem:align-tt}, but the rank does not change before and after the alignment, so we can use the rank before the alignment for the computational complexity.

\subsubsection{$\wedge$BC:}
$\wedge$BC corresponds to the Hadamard product $\mathcal{A}^\prime\circ\mathcal{B}^\prime$ of aligned tensors (by Lemma~\ref{lem:align-tt}) because, for binary values $a$ and $b$, $ab=\min\{a, b\}=a\wedge b$.
Therefore, the complexity equals $\mathcal{O}((m+\ell)r^2q^2)$ by Lemma~\ref{lem:basic-operations-tt}.

\subsubsection{CD:}
For a given tensor train $\mathcal{A}$ and a term $\gamma$ conprised of $\ell$-literals, CD is the transformation constructing a tensor train corresponding to $f|\gamma$.
Following a previous study on SDD~\cite{darwiche2011sdd}, we assume that the input binary variables remain unchanged by CD.
CD can be performed as follows:
The function $f|x_i$ satisfies $(f|x_i)(a_1,\ldots,a_{j-1},0,a_{j+1},\ldots,a_n)=(f|x_i)(a_1,\ldots,a_{j-1},1,a_{j+1},\ldots,a_n)=f(a_1,\ldots,a_{j-1},1,a_{j+1},\ldots,a_n)$ for any assignments $a_k\in\{0,1\}$ $(k=1,\ldots,j-1,j+1,\ldots,n)$.
Thus, the tensor train representation of $f|x_i$ can be constructed by copying $\mathcal{A}_{\pi^{-1}(j)}(\cdot,1,\cdot)$ to $\mathcal{A}_{\pi^{-1}(j)}(\cdot,0,\cdot)$ if $\pi^{-1}(j)$ exists.
Similarly, the tensor train representation of $f|\neg x_i$ can be constructed by copying $\mathcal{A}_{\pi^{-1}(i)}(\cdot,0,\cdot)$ to $\mathcal{A}_{\pi^{-1}(i)}(\cdot,1,\cdot)$ if $\pi^{-1}(j)$ exists.
The tensor train representation of $f|\gamma$ for a term $\gamma$ can be constructed by repeating the above operation for the literals in $\gamma$.
The time complexity is $\mathcal{O}(\ell r^2)$ because the operation of copying $\mathcal{O}(r^2)$ elements is performed $\ell$ times.

\subsubsection{ME:}
If a representation supports CO and CD in polytime, then we can perform ME in polytime~\cite{knowledge}.
Briefly, we enumerate the solutions by creating something resembling a trie-tree, using the divide-and-conquer method.
First, we check whether a model exists for $f$ using CO.
If no model is available, the procedure is completed.
Otherwise, we create two Boolean functions $f \wedge x_{\pi(1)}$ and $f \wedge \neg x_{\pi(1)}$, both of which can be obtained by CD.
We then verify the consistency of both Boolean functions.
If the results are inconsistent, we stop the procedure.
Otherwise, we recursively repeat the process for all $x_{\pi(i)}$.
As CO and CD can be performed in polytime, these substitutions and the existence of models can also be performed in polytime, and the overall computational complexity can be reduced in the tensor train size and the number of models. 
As CO and CD are repeated $\mathcal{O}(M)$ times, the computational complexity is $\mathcal{O}(Mmr^2)$, where $M$ denotes the number of models.

\subsubsection{$\lnot$C:}
As explained in the proof of Lemma~\ref{lem:simple-tt}, $\lnot$C can be calculated by $\mathbf{1}+(-\mathcal{A})$, where $-\mathcal{A}$ is obtained by multiplying the leftmost tensor $\mathcal{A}_1$ in the tensor train by $-1$.
Hereafter, we write $\mathbf{1}+(-\mathcal{A})$ as $\mathbf{1}-\mathcal{A}$.
Because $\mathbf{1}$ can be represented by a tensor train whose number of modes is $l=m$ and rank is $q=1$, the complexity is $\mathcal{O}(mr^2)$, which is derived from the summation in Lemma~\ref{lem:basic-operations-tt}.
The elements of the output tensor train comprise the elements of the original tensor multiplied by $-1$ and arranged in disjoint positions of $1$; thus, it is closed under $\{-1,0,1\}$.

\subsubsection{$\vee$BC:}
$\vee$BC can be calculated by aligning the modes by Lemma~\ref{lem:align-tt} and then computing $\mathbf{1}-((\mathbf{1}-\mathcal{A}^\prime)\circ (\mathbf{1}-\mathcal{B}^\prime))$ because $f \vee g = \lnot\{(\lnot f)\wedge (\lnot g)\}$ by De Morgan's law.
The computational complexity is 
$\mathcal{O}((m+\ell)r^2q^2)$
because the most time-consuming part is taking the logical conjunction, which costs 
$\mathcal{O}((m+\ell)r^2q^2)$. 
Note that, even after taking $\vee$BC, every element in the tensor train representation is closed under $\{-1,0,1\}$ because $\wedge$BC and $\neg$C is also closed under $\{-1,0,1\}$.

\subsubsection{SFO:}
Similar to CD, we assume that the variable set is unchanged by SFO.
SFO can be performed using CD and $\vee$BC.
We only need to compute $(f|x)\vee(f|\neg x)$. The ranks of $f|x$ and $f|\neg x$ are both $r$, and, finally, the rank becomes $\mathcal{O}(r^2)$ by disjunction.
Thus, the computational complexity is $\mathcal{O}(mr^4)$.

\end{document}